\documentclass[11pt,draftcls,onecolumn]{IEEEtran}
\usepackage{bbm,dsfont,mathrsfs,fixmath}
\usepackage{amsmath,amssymb,eucal,graphicx}
\usepackage{stmaryrd,cite}
\usepackage{amsmath,amstext,amsfonts,amssymb}
\usepackage{epsfig}
\usepackage{exscale}
\usepackage{enumerate}

\newtheorem{definition}{Definition}
\newtheorem{lemma}{Lemma}
\newtheorem{cor}{Corollary}[section]

\newtheorem{assumption}{Assumption}[section]
\newtheorem{remark}{Remark}[section]
\newtheorem{property}{Property}[section]
\newtheorem{theorem}{Theorem}

\newcommand{\cond}{\,\vert\,}

\newcommand{\defeq}{\triangleq}

\newcommand{\T}{{\scriptstyle\mathsf{T}}}
\renewcommand{\H}{{\scriptstyle\mathsf{H}}}

\newfont{\bbb}{msbm10 scaled 500}

\newfont{\bb}{msbm10 scaled 1100}
\newcommand{\CC}{\mbox{\bb C}}

\newcommand{\bv}{{\bf b}}

\newcommand{\ev}{{\bf e}}

\newcommand{\gv}{{\bf g}}
\newcommand{\hv}{{\bf h}}

\newcommand{\uv}{{\bf u}}

\newcommand{\vv}{{\bf v}}
\newcommand{\xv}{{\bf x}}
\newcommand{\yv}{{\bf y}}
\newcommand{\zv}{{\bf z}}
\newcommand{\zerov}{{\bf 0}}

\newcommand{\Am}{{\bf A}}
\newcommand{\Bm}{{\bf B}}

\newcommand{\Gm}{{\bf G}}
\newcommand{\Hm}{{\bf H}}
\newcommand{\Id}{{\bf I}}

\newcommand{\Mm}{{\bf M}}

\newcommand{\Sm}{{\bf S}}

\newcommand{\Cc}{{\cal C}}

\newcommand{\Gc}{{\cal G}}
\newcommand{\Hc}{{\cal H}}

\newcommand{\Nc}{{\cal N}}

\newcommand{\Tc}{{\cal T}}

\newcommand{\Phim}{\hbox{\boldmath$\Phi$}}

\newcommand{\Thetam}{\hbox{\boldmath$\Theta$}}

\newcommand{\diag}{{\hbox{diag}}}
\renewcommand{\det}{{\hbox{det}}}
\newcommand{\trace}{{\hbox{tr}}}
\newcommand{\rank}{{\hbox{rank}}}

\DeclareFontFamily{U}{cmfi}{}
\DeclareFontShape{U}{cmfi}{m}{n}{ <-> cmfi10 }{}
\DeclareSymbolFont{CMFI}{U}{cmfi}{m}{n}

\newcommand{\vu}{\pmb{u}}
\renewcommand{\vv}{\pmb{v}}

\newcommand{\taulog}{\mathsf{o}(\log P)}

\newcommand{\vA}{\vv_\textsf{A}}
\newcommand{\vB}{\vv_\textsf{B}}

\newcommand{\nA}{n_\textsf{A}}
\newcommand{\nB}{n_\textsf{B}}

\newcommand{\dA}{d_\textsf{A}}
\newcommand{\dB}{d_\textsf{B}}

\newcommand{\rvW}{{{W}}}
\newcommand{\rvx}{{{x}}}
\newcommand{\rvy}{{{y}}}
\newcommand{\rvz}{{{z}}}
\newcommand{\rvU}{\pmb{{u}}}
\newcommand{\rvV}{\pmb{{v}}}
\newcommand{\rvE}{\pmb{{e}}}
\newcommand{\rvY}{\pmb{{y}}}
\newcommand{\rvZ}{\pmb{{z}}}
\newcommand{\rvX}{\pmb{{x}}}
\newcommand{\rvB}{\pmb{{b}}}
\newcommand{\rvH}{\pmb{{H}}}

\renewcommand{\Thetam}{ {\bf{{\Theta}}}}
\newcommand{\ThetamA}{\Thetam_{\textsf{A}}}
\newcommand{\ThetamB}{\Thetam_{\textsf{B}}}
\renewcommand{\Phim}{ {\bf{{\Phi}}} }
\newcommand{\PhimA}{\Phim_{\textsf{A}}}
\newcommand{\PhimB}{\Phim_{\textsf{B}}}

\newcommand{\He}{\pmb{\mathsf{\widehat{H}}}}
\newcommand{\Ge}{\pmb{\mathsf{\widehat{G}}}}

\newcommand{\WA}{\rvW_\textsf{A}}

\newcommand{\WB}{\rvW_\textsf{B}}
\newcommand{\WAhat}{\hat{\rvW}_\textsf{A}}
\newcommand{\WBhat}{\hat{\rvW}_\textsf{B}}
\newcommand{\dotle}{\ \dot{\le}\  }
\newcommand{\dotge}{\ \dot{\ge}\  }

\renewcommand{\Am}{\pmb{A}}
\renewcommand{\Bm}{\pmb{B}}

\renewcommand{\Gm}{\pmb{G}}
\renewcommand{\Hm}{\pmb{H}}

\renewcommand{\Mm}{\pmb{M}}

\renewcommand{\Sm}{\pmb{S}}

\renewcommand{\bv}{\pmb{b}}

\renewcommand{\ev}{\pmb{e}}

\renewcommand{\gv}{\pmb{g}}
\renewcommand{\hv}{\pmb{h}}

\renewcommand{\uv}{\pmb{u}}
\renewcommand{\vv}{\pmb{v}}

\renewcommand{\xv}{\pmb{x}}
\renewcommand{\yv}{\pmb{y}}
\renewcommand{\zv}{\pmb{z}}

\begin{document}

\title{Secrecy Degrees of Freedom of MIMO Broadcast Channels with
Delayed CSIT}

\author{
\thanks{Manuscript submitted to IEEE Transactions on Information Theory in December 2011.}
Sheng~Yang,~\IEEEmembership{Member,~IEEE,} Mari~Kobayashi,~\IEEEmembership{Member,~IEEE,} 
Pablo~Piantanida,~\IEEEmembership{Member,~IEEE,} \thanks{S. Yang, M. Kobayashi, and P. Piantanida are with the Telecommunications department of SUPELEC,  3 rue Joliot-Curie, 91190 Gif-sur-Yvette, France.~(e-mail: \texttt{\{sheng.yang, mari.kobayashi, pablo.piantanida\}@supelec.fr})}
Shlomo~Shamai~(Shitz),~\IEEEmembership{Fellow,~IEEE} 
\thanks{S. Shamai~(Shitz) is with Technion-Israel Institute of Technology, Haifa, Israel.~(e-mail: \texttt{sshlomo@ee.technion.ac.il})}
\thanks{This work was partially supported by the framework of the FP7 Network of Excellence in Wireless Communications NEWCOM++.}
}

\maketitle

\begin{abstract}
The degrees of freedom~(DoF) of the two-user Gaussian multiple-input and multiple-output~(MIMO) broadcast channel with confidential message~(BCC) is studied under the assumption that delayed channel state information (CSI) is available at the transmitter.  We characterize the optimal secrecy DoF~(SDoF) region and show that it can be achieved by a simple artificial noise alignment (ANA) scheme. The proposed scheme sends the confidential messages superposed with the artificial noise over several time slots. Exploiting delayed CSI, the transmitter aligns the transmit signal in such a way that the useful message can be extracted at the intended receiver but is completely drowned by the artificial noise at the unintended receiver. The proposed scheme can be interpreted as a non-trivial extension of Maddah-Ali Tse (MAT) scheme and enables us to quantify the resource overhead, or equivalently the DoF loss, to be paid for the secrecy communications.
\end{abstract}
\newpage

\section{Introduction}
We consider the two-user Gaussian multi-input multi-output broadcast channel with confidential messages (MIMO-BCC), where the transmitter sends two confidential messages to receivers~\textsf{A}~and~\textsf{B}, respectively,  while keeping each of them secret to the unintended receiver.  By letting $m$, $\nA$, and $\nB$ denote the number of antennas at the transmitter, receiver~\textsf{A}, and receiver~\textsf{B}, respectively, the corresponding channel outputs are given by
\begin{subequations}
\begin{align}
    \yv_{t} &= \Hm_{t} \xv_{t} + \ev_{t}, \\
    \zv_{t} &= \Gm_{t}  \xv_{t} + \bv_{t}, \quad t=1, 2, \ldots, n,
\end{align}
\end{subequations}
where $(\yv_{t}, \zv_{t})$ denotes the observations
at the receiver \textsf{A} and \textsf{B}, respectively,  at time
instant $t$; $\Hm_t\in \Hc \subseteq \CC^{\nA\times m},\Gm_{t}\in\Gc
\subseteq \CC^{\nB\times m}$ are the associated channel matrices;
$(\ev_{t},\bv_{t})$ are assumed to be independent and identically
distributed (i.i.d.) additive white Gaussian noises
$\sim\Nc_{\Cc}(\zerov,\Id)$; the input vector $\xv_t\in\CC^{m\times 1}$
is subject to the average power constraint 
\begin{align}
  \label{PowerConstraint} \frac{1}{n}\sum\limits_{t=1}^n \trace (
  \xv_{t}  \xv_{t}^\H  ) \leq P.  
\end{align}  
Furthermore, as in \cite{maddah2010degrees}, we assume any arbitrary
stationary fading process such that  $\Hm_{t}$ and $\Gm_{t}$ are
mutually independent and change from an instant to another one in an
independent manner. Note that the channel at hand boils down to the
conventional Gaussian MIMO wiretap channel  where the transmitter wishes
to send one message to the intended receiver while keeping it secret to the other one, namely, the eavesdropper. 

The secrecy capacity region of the two-user MIMO Gaussian BCC with
perfect channel state information at transmitter~(CSIT)and receivers
has been characterized in \cite{liu2010multiple} (see also references
therein). As a special case, the Gaussian MIMO wiretap channel has been extensively studied in
\cite{liu2009note,khisti2010secure,khisti2010secure2,oggier2011secrecy,immse09}.
However, the secrecy capacity of the MIMO Gaussian wiretap channel with
general~(imperfect) CSI at the transmitter remains open. Since a complete characterization
of the capacity region in this case is very difficult~(if not impossible), a number of 
contributions have focused on the so-called secrecy degrees of freedom~(SDoF), by
capturing the behavior in high signal-to-noise~(SNR) regime~(see
\cite{yingbin2009compound, khisti2011interference,
kobayashi2009compound, kobayashi2010secrecy} and references therein).
References \cite{yingbin2009compound, khisti2011interference,
kobayashi2009compound} investigated the compound models where channel
uncertainty at the encoder is modeled as a set of finite channel states, while 
\cite{kobayashi2010secrecy} investigated the scenario
where the transmitter knows some temporal structure of the
block-fading processes. A fundamental
observation is that unless two channels enjoy asymmetric 
statistical properties\footnote{This may be in terms of asynchronous
fading variation, different fading speed, number of antennas, etc..},
the perfect secrecy cannot be guaranteed under a general CSIT assumption. 
In other words, if the statistics of the underlying channels seen by
both receivers are symmetrical, additional side information~(not
necessarily instantaneous CSIT) is essential to ensure a positive SDoF,
by introducing some asymmetry at the encoder. As a matter of fact, this
reveals one of the major limitations of the wiretap model whose
performance strongly depends on the quality of the channel state
information at the transmitter side. Evidently, theoretically addressing
CSI issues is of fundamental impact for secrecy systems. 

Recently, in the context of multi-antenna broadcast channel, the
pioneering work~\cite{maddah2010degrees} showed that completely outdated
channel state information at the transmitter is still very useful and
increases the degrees of freedom of the multi-user channel.  Motivated
by this exciting result, the new assumption, commonly referred to as
delayed CSIT, has since been applied to several multi-user settings,
including the MIMO broadcast channel, X channel, and interference
channel~\cite{vaze2010degrees,maleki2010retrospective,vaze2011degrees,abdoli2011degrees}.
Non-trivial gain of degrees of freedom have been shown in all these
settings with delayed CSIT. The main idea behind the utility of delayed
CSIT can be best described with the term ``retrospective interference
alignment'' introduced in \cite{maleki2010retrospective} and
\cite{FandT_Jafar}. That is, the knowledge of causal channel state is
used to align the interference between users into a
spatial/temporal subspace with a reduced dimension at each receiver. 

In this paper, we study the impact of delayed CSIT on the secrecy
degrees of freedom in a MIMO broadcast channel. 
In our setting, delayed CSI of a given receiver is available both at the transmitter and the other
receiver\footnote{Unless it is explicitly mentioned, we assume that
delayed CSI of both channels is available to the transmitter, i.e., it
observes $\Hm^{t-1}$ and $\Gm^{t-1}$ for every $t=1,2,\dots$.}, whereas
each receiver knows its own instantaneous channel. Such a scenario is of
practical interest since the receivers may send their channel states to
the transmitter via delayed feedback links that may be
overheard by the other receivers. 
We first characterize the optimal SDoF of the Gaussian MIMO wiretap
channel with delayed CSIT. It is shown that delayed CSIT can
significantly improve the SDoF, provided that the number of transmit
antennas is larger than that of receive antennas, i.e.,
$m>\max(\nA,\nB)$. In this case, we prove that a simple artificial noise
alignment (ANA) scheme achieves the optimal SDoF. The proposed scheme
sends the confidential symbols embedded by the artificial noise in such a way
that the artificial noise is aligned in a subspace at the legitimate
receiver while it fills the full signal space at the eavesdropper.  
The case of partial knowledge where the transmitter has delayed CSI only on
the legitimate channel is also investigated. In this case, we show that
a strictly smaller SDoF is achieved compared to the case with delayed
CSIT on both channels. Then, we consider the two-user Gaussian MIMO-BCC
and characterize the optimal SDoF region. The achievability follows from
an artificial noise alignment scheme adapted to convey two confidential messages. The proposed
scheme can be seen as a non-trivial extension of the
Maddah-Ali~Tse~(MAT) scheme. A simple comparison with the MAT scheme enables
us to quantify the resource overhead, or equivalently the DoF loss, to
be paid to guarantee the confidentiality of messages. Although delayed
CSIT is found beneficial for a large range of transmit antennas analogy
to the conclusions drawn for other network systems without secrecy
constraints
\cite{maddah2010degrees,maleki2010retrospective,vaze2011degrees}, we
remark that the lack of perfect CSIT significantly degrades the
performance of the secrecy systems. 

The rest of the paper is organized as follows. Section
\ref{section:Preliminaries} introduces the assumptions and some useful
lemmas while Section~\ref{section:MainResults} summarizes our main results
on the optimal SDoF. Sections \ref{section:Proof-wiretap}~and~\ref{section:Proof-bcc} are devoted to
proof of the main theorems. Finally, the paper is concluded in Section
\ref{section:Conclusions} with some open problems and future perspectives.

\section{Notations, Definitions, and Assumptions} \label{section:Preliminaries}
\subsection{Notations}
Boldface lower-case letters $\vv$ and upper-case letters $\Mm$ are used
to denote vectors and matrices, respectively.
We use the superscript notation $X^n$ to
denote a sequence $(X_1,\ldots,X_n)$ for any type of variables.  Matrix
transpose, Hermitian transpose, inverse, trace, and determinant are
denoted by $\Am^\T$, $\Am^{\H}$, $\Am^{-1}$, $\trace( \Am)$, and
$\det(\Am)$, respectively. We let $\diag(\{\Am_t\}_{t})$ denote the
block diagonal matrix with the matrices $\Am_t$ as diagonal elements. Logarithm
is in base $2$ unless otherwise is specified. The differential entropy
of $X$ is denoted by $h(X)$. $(x)^+$ means $\max\left\{ 0, x
\right\}$. The little-o notation $\taulog$ stands for any real-valued function $f(P)$ such that
${\displaystyle \lim_{P\to \infty}} \frac{f(P)}{\log P}=0$. The dot equality means 
the equality on the ``pre-log'' factor, i.e., $f(P) \doteq g(P)$ is
equivalent to  
$f(P) = g(P) + \taulog$; the dot inequalities $\dot{\ge}$ and $\dot{\le}$
are similarly defined. 

\subsection{Assumptions and Definitions}

The following assumptions and definitions will be applied in the rest of
the paper. 

\begin{definition}[channel states]
  The channel matrices $\Hm_t$ and $\Gm_t$ are called the states of the
  channel at instant $t$. For simplicity, we also define the state
  matrix $\Sm_t$ as $\Sm_t = \left[ \begin{smallmatrix} \Hm_t \\ \Gm_t
  \end{smallmatrix} \right]. $
\end{definition}

\begin{assumption}[delayed CSIT]\label{assuption:CSI}
At each time $t$, the states of the past $\Sm^{t-1}$ are known to all terminals. However, the instantaneous states $\Hm_t$ and $\Gm_t$ are only known
to the respective receivers. 
\end{assumption}

Under these assumptions, we define the code and the optimal SDoF region summarized below. 
\begin{definition}[code and SDoF region] A code for the Gaussian
  MIMO-BCC with delayed CSIT consists of:
\begin{itemize}
\item A sequence of stochastic encoders given by 
  $$
  \{F_t: \mathcal{W}_{\textsf{A}} \times \mathcal{W}_{\textsf{B}} \times
  \mathcal{H}^{t-1} \times \mathcal{G}^{t-1} \longmapsto
  \CC^{m}\}_{t=1}^n,
  $$ 
  where the messages $\WA$ and $\WB$ are uniformly distributed over
  $\mathcal{W}_{\textsf{A}}$ and $\mathcal{W}_{\textsf{B}}$, respectively. 
\item The decoder A is given by the mapping $\WAhat: \mathbb{C}^{\nA
  \times n} \times \mathcal{H}^n \times  \mathcal{G}^{n-1} \longmapsto
  \mathcal{W}_{\textsf{A}}$.
\item The decoder B is given by the mapping $\WBhat: \mathbb{C}^{\nB
  \times n} \times \mathcal{H}^{n-1} \times  \mathcal{G}^n \longmapsto \mathcal{W}_{\textsf{B}}$.
\end{itemize}
A SDoF pair $(\dA,\dB)$ is said {\it achievable} if there exists a code that satisfies 
 the reliability conditions at both receivers
\begin{align}
  & \lim_{P\to \infty} \liminf_{n\to \infty} \frac{\log
  |\mathcal{W}_{\textsf{A}}(n,P)|}{n\log P}\geq \dA,
  \;\;\lim_{P\to\infty}\limsup\limits_{n\to \infty} \Pr\left\{ \WA\neq \hat{\WA}\right\}= 0, \\
  & \lim_{P\to \infty} \liminf_{n\to \infty} \frac{\log
  |\mathcal{W}_{\textsf{B}}(n,P)|}{n\log P}\geq \dB,
  \;\;\lim_{P\to\infty}\limsup\limits_{n\to \infty} \Pr\left\{ \WB\neq \hat{\WB}\right\}= 0, 
\end{align}
as well as the perfect secrecy condition 
\begin{align}  \label{eq:Constraint1-bcc}
& \lim_{P\to \infty} \limsup_{n\to \infty} \frac{I(\WA;\rvZ^n, \Sm^{n})}{n\log P}=0,\\  \label{eq:Constraint2-bcc}
& \lim_{P\to \infty} \limsup_{n\to \infty} \frac{I(\WB;\rvY^n, \Sm^{n})}{n\log P}=0. 
\end{align}
The union of all achievable pairs $(\dA,\dB)$ is called the optimal SDoF region. 
\end{definition}

\begin{assumption}[channel symmetry]\label{assumption:independency}  
 At any instant $t$, the rows of the state matrix $\Sm_t$ are
 independent and identically distributed. Furthermore, we limit ourselves to the class of fading processes in which  
 the state matrix $\Sm_t$ has full rank $\min\left\{
m, \nA+\nB \right\}$ almost surely at any time instant $t$.\footnote{This assumption is used to prove the achievability although
the converse proof does not need such an assumption.} 
\end{assumption}

As direct consequences of the channel symmetry, we readily have that the marginal distributions of any antenna output are equal conditioned on the same previous observations and/or the source message. Namely, we have the following property. 
\begin{property}[channel output symmetry]\label{property:symmetry}
Let $\Omega_t = \left\{ \rvy_{1,t}, \ldots, \rvy_{\nA,t},
\rvz_{1,t}, \ldots, \rvz_{\nB,t} \right\}$ be the collection of random variables representing all
antenna outputs at time instant $t$. Then, for any subset
$\omega_{\mathcal{J}}$ and $\omega_{\mathcal{K}}$ of random variables in 
$\Omega_t$ satisfying $|\omega_{\mathcal{J}}|=|\omega_{\mathcal{K}}|$, we have 
\begin{align}
  \Pr(\omega_{\mathcal{J}} \cond \rvY^{t-1}, \rvZ^{t-1},U_t) &= \Pr(\omega_{\mathcal{K}} \cond \rvY^{t-1}, \rvZ^{t-1}, U_t)
\end{align}
for any random variables $U_t \leftrightarrow (\Hm_t,\Gm_t, W) \leftrightarrow \Omega_t$ with $t=\{1,\dots,n\}$   that  form a Markov chain.  
\end{property}

Using the fact that current channel outputs do not depend on the future channel
realizations, we can easily show that Property \ref{property:symmetry}
also holds when we add the conditioning on $\Sm^n$, namely,
\begin{align}
  h(\omega_{\mathcal{J}} \cond \rvY^{t-1}, \rvZ^{t-1}, \Sm^n,\rvW)  &= h(\omega_{\mathcal{K}} \cond \rvY^{t-1}, \rvZ^{t-1}, \Sm^n,\rvW).
\end{align}
In the following, we omit the conditioning on $\Sm^n$ for notation simplicity.

\subsection{Preliminaries}
For sake of clarity, we collect the results that will be used repeatedly in the rest of the paper. First, the following lemma is the direct consequences of the channel output
symmetry.  
\begin{lemma}[properties of channel symmetry]\label{lemma:jointvssingle}
The following inequalities hold under the {\it channel output symmetry} Property \ref{property:symmetry}:  
\begin{subequations}
  \begin{align}
    \min\{m,\nA+\nB\}\, h(\rvZ^n) &\dotge \nB\, h(\rvY^n ,\rvZ^n) \label{eq:tmp1},\\
    \min\{m,\nA+\nB\}\, h(\rvY^n) &\dotge \nA\, h(\rvY^n ,\rvZ^n),  \\ 
    \min\{m,\nA+\nB\}\, h(\rvZ^n) &\dotge \nB\, h(\rvY^n), \label{eq:tmp822}\\ 
       \min\{m,\nA+\nB\}\, h(\rvY^n) &\dotge \nA\, h(\rvZ^n). \label{eq:tmp823}
  \end{align}%
 \end{subequations}
  Furthermore, same inequalities hold true conditional on $W$. 
\end{lemma}
\begin{proof}
  The first two inequalities are proved in Appendix
  \ref{appendix:lemma-jointvssingle}. 
To prove \eqref{eq:tmp822}, from \eqref{eq:tmp1}, we have
  \begin{align}
    h(\rvZ^n) & \dotge \frac{\nB}{\min\{m,\nA+\nB\}} h(\rvY^n, \rvZ^n) \\
    &\ge \frac{\nB}{\min\{m,\nA+\nB\}} h(\rvY^n),
  \end{align}%
  where the last inequality
  comes from the fact that $h(\rvZ^n \cond
  \rvY^n)\geq h(\rvZ^n \cond
  \rvY^n, \rvX^n) = h(\rvB^n) = \taulog$. Same steps can be applied to obtain
  \eqref{eq:tmp823}.  
\end{proof}
Then, all the achievable DoF results are essentially based on the rank of the channel
matrices. 
\begin{lemma}
For any matrix $\Am$ which does not depend on $P$, we have
\begin{align}
\lim_{P\to \infty} \frac{\log \det(\Id + P \Am\Am^\H)}{\log P} = \rank(\Am). 
\end{align}
\end{lemma}\vspace{2mm}

\begin{proof}
  Let $(\sigma_1,\ldots,\sigma_r)$ be the $r\defeq\rank(\Am)$ non-zero
  singular values of $\Am$. Then, we have that 
  $$
  \log\det(\Id+P\Am\Am^\H)=
  \sum\limits_{k=1}^r \log(1+P\sigma_k^2) \doteq r \log P,
  $$ 
  since the non-zero singular values do not depend on $P$ and thus do not vanish with $P$ either. 
\end{proof}

\section{Main Results}\label{section:MainResults}

In this section, we highlight our main results on the optimal SDoF of
the Gaussian MIMO wiretap channel and then on the more general Gaussian
MIMO broadcast channel with confidential messages. We shall interpret the results through comparisons and numerical evaluations.
 \subsection{Wiretap Channel}
 
\begin{theorem}[wiretap channel with delayed CSIT]\label{thm:SDoF}
In presence of delayed CSIT on both the legitimate channel and the eavesdropper channel, the optimal SDoF of the Gaussian MIMO wiretap channel with $m, \nA, \nB$ antennas at the transmitter, the legitimate receiver, the eavesdropper, respectively, is given by 
\begin{equation}\label{SDoF-H}
  d_s(\nA,\nB,m)= \begin{cases}
0, &  m \leq \nB,\\
\displaystyle{m-\nB}, & \nB < m \leq  \nA,\\ 
\displaystyle{\frac{\nA m(m-\nB)}{\nA\nB + m(m-\nB)}}, &  \max\{\nA,\nB\} < m \leq {\nA + \nB},\\
\displaystyle{\frac{\nA (\nA+\nB)}{ \nA+2\nB}}, & m> \nA+\nB.
\end{cases}
\end{equation}
\end{theorem}
\vspace{4mm}

In the wiretap setting, it is not always reasonable to assume any CSI on
the eavesdropper channel at the transmitter side. In this case, we may
consider delayed CSIT only on the legitimate channel and without CSIT on
the eavesdropper channel. With this asymmetric CSI assumption, hereafter
referred to as delayed partial CSIT, we can show that a strictly positive
yet smaller SDoF than delayed CSIT on both channels is still achievable
for a wide range of number of antennas. 

\begin{theorem}[wiretap channel with delayed partial CSIT] \label{thm:asymmetric}
In presence of delayed partial CSIT, either on the legitimate channel or the eavesdropper channel, the following SDoF is achievable for MIMO Gaussian wiretap channel for $m>\max\{\nA, \nB\}$
\begin{equation}
  d_s^{\rm partial}(\nA,\nB,m)= \begin{cases}
\displaystyle{\frac{\nA (m-\nB)}{m}}, &  \max\{\nA,\nB\} < m \leq {\nA + \nB}, \\
\displaystyle{\frac{\nA^2}{ \nA+\nB}}, & m> \nA+\nB.
\end{cases}
\end{equation}
\end{theorem}
\vspace{4mm}
Note that it is the best known achievable SDoF in this setting, although the converse is yet to be proved.

In order to quantify the benefit of delayed CSIT, we summarize the SDoF
with perfect, delayed and without CSIT in Table \ref{table:SDoF} and
provide an example with $\nA=3, \nB=2$ in Fig. \ref{fig:SDoF}. We remark
that delayed CSIT is beneficial only when the number of transmit
antennas is larger than the number of receive antennas, i.e.,
$m>\max\{\nA,\nB\}$, since the SDoF is $(m-\nB)^+$ for $m\leq
\max\{\nA,\nB\}$ with perfect, delayed, and without CSIT.  As the number
$m$ of transmit antennas increases, the SDoF grows until $m=\nA+\nB$ for
perfect and delayed CSIT while it does not increase with $m$ without
CSIT.  It appears that with both perfect and delayed CSIT, we cannot
exploit any gain for $m$ beyond $\nA+\nB$. Furthermore, we remark that
delayed CSIT only on the legitimate channel incurs a non-negligible loss
compared to delayed CSIT on both channels. This is because the transmitter without CSI on the eavesdropper channel cannot access to the signal overheard by the eavesdropper, which reduces the signal dimension to be exploited by the legitimate receiver.

\begin{table}[h]
\caption{Comparison of the SDoF under different CSIT assumptions for $m> \max\{\nA,\nB\} $.}
\centering
\begin{tabular}{|c|c|c|}
  \hline
CSIT &  $ \max\{\nA,\nB\}< m < \nA+ \nB $ & $m\geq  \nA+\nB$\\ \hline
perfect & $m-\nB$ & $\nA$\\
delayed &$ \frac{\nA m(m-\nB)}{\nA\nB + m(m-\nB)}$ & $\frac{\nA (\nA+\nB)}{ \nA+2\nB}$\\
delayed partial & $\frac{\nA (m-\nB)}{m}$ & $\frac{\nA^2}{ \nA+\nB}$ \\
no & $(\nA-\nB)^+$ &  $(\nA-\nB)^+$  \\
  \hline
\end{tabular}
\label{table:SDoF}
\end{table}
\begin{figure}[h]
  \centering
\includegraphics[width=0.6\textwidth,clip=]{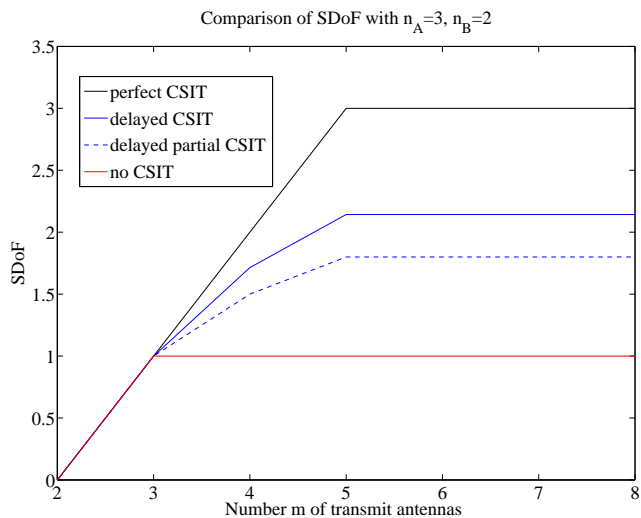}
\caption{SDoF with $\nA=3$ and $\nB=2$ with perfect, delayed, and no CSIT.}
\label{fig:SDoF}
\end{figure}

\subsection{Broadcast channel with confidential messages}

Next, we present the achievable SDoF region of the two-user MIMO-BCC
with delayed CSIT. 
\begin{theorem}[BCC with delayed CSIT]\label{thm:ubBCC}
The optimal SDoF region $\mathcal{R}_{\textrm{BCC}}$ of the two-user
MIMO-BCC with delayed CSIT is given as a set of non-negative
$(\dA,\dB)$ satisfying
\begin{subequations}
\begin{align} 
\frac{\dA}{d_s(\nA,\nB,m)} + \frac{\dB}{ \min\{m, \nA+\nB\}} \leq 1, \label{BCC-inq1}\\
\frac{\dA}{ \min\{m, \nA+\nB\}} + \frac{\dB}{d_s(\nB,\nA,m)} \leq 1, \label{BCC-inq2}
\end{align}  
\end{subequations}
for any $m>\max\{\nA,\nB\}$. If $\nB<m\leq \nA$, we have $\dA\leq m-\nB$
and $\dB=0$, whereas if $\nA<m\leq \nB$, we have  $\dA=0$ and $\dB\leq m-\nA$. 
\end{theorem}
\begin{cor}
For the case $m>\max\{\nA,\nB\}$, the SDoF region is characterized by
the two corner points $(0,d_s(\nB,\nA,m))$, $(d_s(\nA,\nB,m),0)$ and the sum SDoF point given by 
\begin{align}
 (\dA, \dB) =  \begin{cases}
\displaystyle{\left(\frac{\nA(m-\nB)}{m},\frac{\nB(m-\nA)}{m}\right)},&  \max\{\nA,\nB\} < m \leq {\nA + \nB} \\
\displaystyle{\left(\frac{\nA^2}{\nA+\nB},\frac{\nB^2}{\nA+\nB}\right)}, & m> \nA+\nB.
\end{cases}
\end{align}
\end{cor}
\vspace{3mm}

\begin{remark}
We can find trivial upper bounds to the above SDoF region for the case
of $m>\max\{\nA,\nB\}$. On one hand, the SDoF region with delayed CSIT
is dominated by the SDoF region with perfect CSIT. 
The SDoF region with perfect CSIT is square connecting three corner
points $(\min\{\nA,m-\nB\},0)$, $(\min\{\nA,m-\nB\},\min\{m-\nA,\nB\})$,
and $(0,\min\{m-\nA,\nB\})$. We can also compare the above SDoF region
with delayed CSIT and the DoF region of the two-user MIMO-BC with delayed constraint \cite{vaze2010degrees}, given by 
\begin{subequations}
\begin{align} 
\frac{\dA}{\min\{m,\nA\}} + \frac{\dB}{\min\{m, \nA+\nB\}} &\leq 1,\\
\frac{\dA}{\min\{m, \nA+\nB\}} + \frac{\dB}{\min\{m,\nB\}} &\leq 1. 
\end{align}  
\end{subequations}
Obviously, since SDoF is always upper bounded by DoF of the MIMO channel, namely
$d_s(\nA, \nB, m)\leq  \min\{\nA,m\}$ and $d_s(\nB, \nA, m) \leq \min\{\nB, m\}$, the SDoF region is dominated by the DoF region. 
\end{remark}

We provide an insight to the proposed artificial noise alignment scheme which achieves the sum SDoF point  $\left(\frac{1}{2},\frac{1}{2}\right)$ over the two-user MISO-BCC. Let us consider the four-slot scheme where the transmitter sends six independent Gaussian distributed symbols  $\vu \defeq [u_1\ u_2]^\T$, $\vA \defeq
[v_{11} \ v_{12}]^\T$, $\vB \defeq [v_{21} \ v_{22}]^\T$ whose powers
scale equally with $P$. Specifically, the transmit vectors are
given by 
\begin{align}\label{4slot-BCC}
\xv_1 = \vu, \;\; \xv_2 = \vA+ 
\begin{bmatrix}
 \hv_1^{\T} \uv \\ 0 \end{bmatrix}, \xv_3 =  \vB  + \begin{bmatrix}\gv_1^{\T} \uv \\0 \end{bmatrix},
   \xv_4 =\begin{bmatrix}
   (\gv_2 ^\T \vA + g_{21}\hv_1^{\T} \uv) + (\hv_3 ^\T \vB + h_{31}\gv_1^{\T} \uv)\\
    0
    \end{bmatrix},
\end{align}
where, for simplicity of demonstration, we omit the scaling factors that
fulfill the power constraint \eqref{PowerConstraint}. Note that this
simplification, also adopted in \cite{maddah2010degrees} and other
related works, does not affect the high SNR analysis carried out here. 
The following remarks are in order. First, it can be easily shown that,
at receiver~\textsf{A}, $\vA$ lies in a two-dimensional subspace, while
the unintended signal $\vB$ plus the artificial noise are aligned in
another two-dimensional subspace. Thus, the intended message can be
recovered through $\vA$ from the four-dimensional observation at
receiver~\textsf{A}. Second, $\vA$ is drowned in the observation at
receiver~\textsf{B}. More precisely, at
receiver~\textsf{B}, $\vA$ is squeezed into a one-dimension
subspace filled with artificial noise, which makes it impossible to
recover any useful information about \textsf{A}.  
Due to the symmetry, the same holds for $\vB$.
Therefore, we can send simultaneously two confidential symbols to each receiver over
four slots, yielding the sum SDoF point $\left(\frac{1}{2},\frac{1}{2}\right)$. 

The four-slot scheme contains two special cases of interest. If we consider
the MISO wiretap channel where the transmitter wishes to convey $\vA$ to
receiver~\textsf{A} while keeping it secret to receiver~\textsf{B}, we let $\vB=\zerov$
and ignore the third time slot. This provides a SDoF of
$\frac{2}{3}$. If we consider the two-user MISO-BC without secrecy
constraint, we remove the artificial noise transmission by letting
$\uv=\zerov$ and ignoring the first time slot. This boils down to the MAT
scheme~\cite{maddah2010degrees}. The four-slot scheme as well as the more
general artificial noise alignment scheme presented in Section
\ref{section:Proof-bcc} is indeed a non-trivial
extension of retrospective interference alignment schemes for MIMO
broadcast channels~\cite{maddah2010degrees,vaze2010degrees} to 
secure communications. The comparison with the three-slot MAT scheme can
be interpreted as follows. The messages can be kept secret at a price of
an additional resource (one slot), which appears as a DoF loss with
respect to the communication systems without secrecy constraint. 

\begin{figure}
 \centering
\includegraphics[width=0.5\textwidth,clip=]{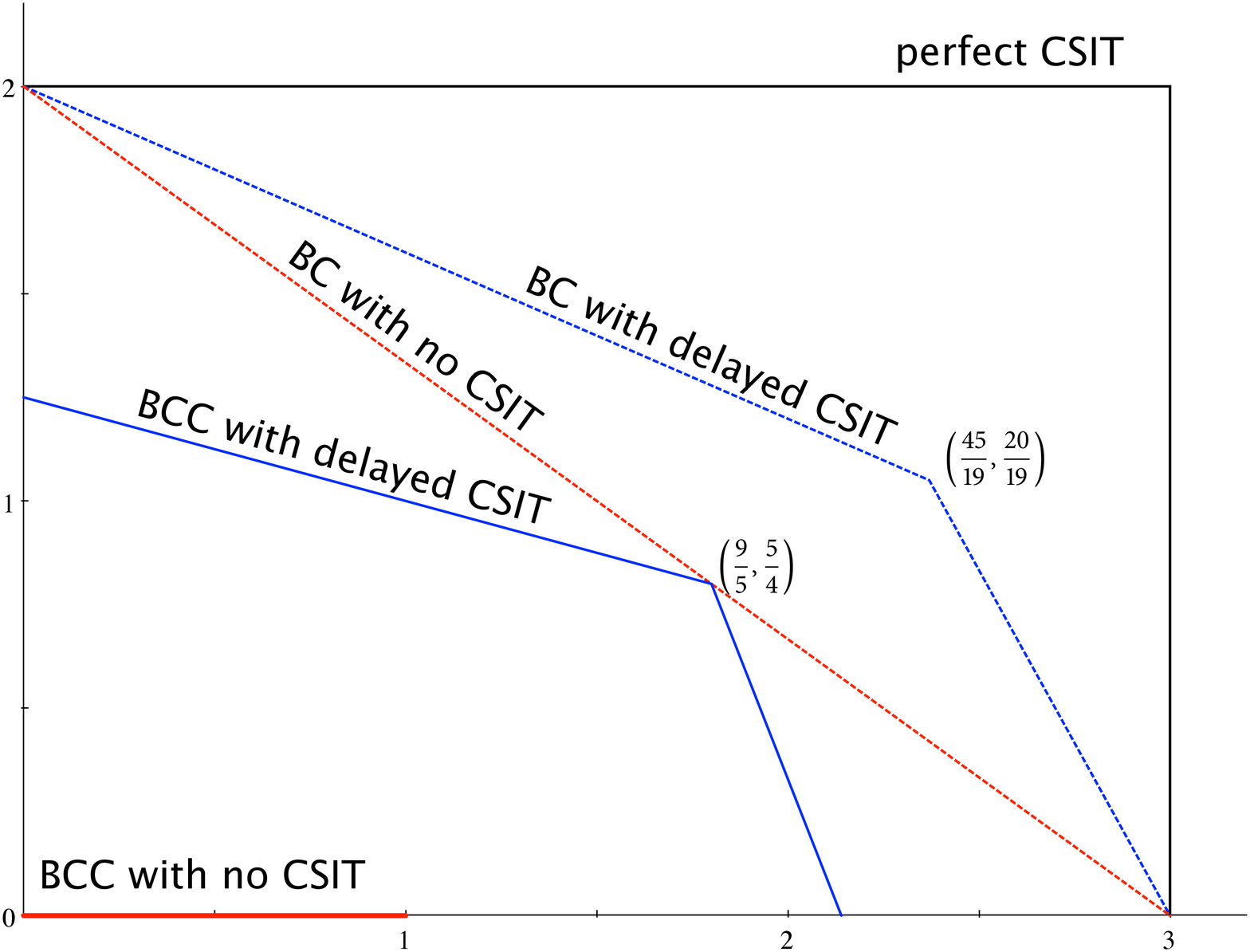}
\caption{The two-user DoF/SDoF region with $m=5, \nA=3, \nB=2$.}
\label{fig:Region}
\end{figure}

In order to visualize the DoF loss due to the secrecy constraints, we
provide an example of the achievable DoF/SDoF regions with $m=5$,
$\nA=3$, and $\nB=2$ in Fig.~\ref{fig:Region}. For the case of perfect
CSIT, the SDoF
region and the DoF region are square. In the MIMO-BC, we send
$(\nA^2(\nA+\nB), \nB^2(\nA+\nB))=(45, 20)$ private symbols to
receiver~\textsf{A} and \textsf{B}, respectively, over a duration of $\nA^2+\nB^2+\nA\nB=19$ slots,
yielding the DoF $\left(\frac{45}{19}, \frac{20}{19}\right)$, as shown
in \cite{vaze2010degrees}. Under the perfect secrecy constraints, we need
an extra phase of the artificial noise transmission of $\nA\nB=6$ slots
to convey two streams securely. This yields the SDoF of
$\left(\frac{9}{5}, \frac{4}{5}\right)$. The comparison with the DoF
region of the MIMO-BC can be interpreted in either an optimistic or a
pessimistic way. On one hand, the benefit of delayed CSIT is more
significant for the SDoF region. On the other hand, we also observe that
the lack of accurate CSIT decreases substantially the SDoF, which
implies  that the secure communications are very sensitive to the
quality of CSIT.

\section{Wiretap Channel: Proofs of Theorems \ref{thm:SDoF} and \ref{thm:asymmetric}}
\label{section:Proof-wiretap}
\subsection{Converse proof of Theorem \ref{thm:SDoF}}

We are now ready to provide the converse by considering different cases below.
\subsubsection{Case $m\leq \nB$} 
From Fano's inequality and the secrecy constraint, we have
\begin{align}
	 n(R-\taulog ) 	 &\le I(\rvW;\rvY^n ) - I(\rvW;\rvZ^n ) \\ 
	 &= I(\rvW;\rvY^n \cond \rvZ^n) - I(\rvW;\rvZ^n \cond \rvY^n) \\
	 &= h(\rvY^n \cond \rvZ^n) - h(\rvY^n  \cond \rvZ^n, \rvW) - I(\rvW;\rvZ^n \cond \rvY^n) \\
	 &\le  h(\rvY^n \cond \rvZ^n)  \label{drop} \\	
         &\le  \sum_{t=1}^n h(\rvY_t \cond \rvZ_t) \label{eq:tmp777}  \\	
         &=  \sum_{t=1}^n h(\rvY_t - \rvH_t\,\hat{\rvX}_{\rvZ,t} \cond
         \rvZ_t) \label{translation}\\ 	
         &= \sum_{t=1}^n h(\rvE_t - \rvH_t\,(\hat{\rvX}_{\rvZ,t}-\rvX_t)  \cond \rvZ_t) \label{translation2}\\ 	
         &= \sum_{t=1}^n h(\rvE_t - \rvH_t\,(\hat{\rvX}_{\rvZ,t}-\rvX_t)) \label{non-conditioning}\\
	 &=n\, \taulog,  \label{eq:tmp888} 
\end{align}
where (\ref{drop}) is from the fact that both $I(\rvW;\rvZ^n \cond \rvY^n)$ and
$h(\rvY^n \cond \rvZ^n, \rvW)$ are non-negative; in \eqref{translation} we use the
fact that translations preserve differential entropy and let $\hat{\rvX}_{\rvZ,t}$
denote the MMSE estimation of $\rvX_t$ given $\rvZ_t$; the last equality
holds because the estimation error does not scale with $P$. 

\subsubsection{Case $ \nB < m\leq \max\left\{ \nA, \nB \right\}$} 
Since this case happens only when $\nB < m \le \nA$, we can assume $m
\leq \nA$. Starting from \eqref{drop}, we have  
\begin{align}
	 n(R-\taulog) 
         &\le h(\rvY^n \cond \rvZ^n) \\          
         &\dotle\! \frac{\min\{m, \nA+\nB \}-\nB}{\nB} h(\rvZ^n) \label{eq:tmp673}\\
         & \leq n (m-\nB) \log P + n\,\taulog,
\end{align}
where \eqref{eq:tmp673} follows straightforwardly from \eqref{eq:tmp1};
the last inequality comes from the fact that i.i.d. Gaussian variables
maximize the differential entropies under the variance constraint.

\subsubsection{Case $m>\max\{\nA,\nB\}$}  
In the following we let $\tilde{m}=\min\{m, \nA+\nB\}$ for notation simplicity. 
We remark that two upper bounds can be obtained as a direct consequence of Lemma \ref{lemma:jointvssingle}. 
One one hand, \eqref{eq:tmp673} still holds 
\begin{align} 
  I(\rvW;\rvY^n) - I(\rvW;\rvZ^n) 
  &\dotle\! \frac{\tilde{m}-\nB}{\nB} h(\rvZ^n).\label{appli}
\end{align}
On the other hand, we have 
\begin{align}
 I(\rvW;\rvY^n) - I(\rvW;\rvZ^n) &= h(\rvY^n) -h(\rvY^n\cond \rvW) -
 h(\rvZ^n) + h(\rvZ^n\cond \rvW)  \\
  &\dotle\!  h(\rvY^n) + \left(1-\frac{\nA}{\tilde{m}} \right)h(\rvZ^n\cond
  \rvW) - h(\rvZ^n) \label{eq:ineq-lemma2}\\  
   &\leq h(\rvY^n)  - \frac{\nA}{\tilde{m}} h(\rvZ^n), \label{eq:tmplast}
\end{align}
where \eqref{eq:ineq-lemma2} follows from \eqref{eq:tmp823};
\eqref{eq:tmplast} follows from $h(\rvZ^n \cond W) \le h(\rvZ^n)$;  
By combining the above two upper bounds, we readily have 
\begin{align}
  n(R-\taulog ) &\le I(\rvW;\rvY^n ) - I(\rvW;\rvZ^n) \\
    &\dotle\!\! \min\left\{ \frac{\tilde{m}-\nB}{\nB}h(\rvZ^n  ), h(\rvY^n ) - \frac{\nA}{\tilde{m}}
    h(\rvZ^n ) \right\} \label{min} \\
    & \leq \max_{\alpha} \max_{\beta}  \min\left\{
    \frac{\tilde{m}-\nB}{\nB}\beta, \alpha - \frac{\nA}{\tilde{m}}
    \beta \right\} \label{maxmax}\\
    & \leq \max_{\alpha}  {\alpha}
    \left({1+\frac{\nA\nB}{\tilde{m}(\tilde{m}-\nB)}}\right)^{-1} \label{innerprob}\\
  &\dotle   
  \left({1+\frac{\nA\nB}{\tilde{m}(\tilde{m}-\nB)}}\right)^{-1}
  {\nA} n \log P,
\end{align}
where \eqref{maxmax} is because the RHS of \eqref{min} can only increase by maximizing it over both entropies $\alpha= h(\rvY^n ) , \beta= h(\rvZ^n ) $; in \eqref{innerprob} the inner maximization is
solved by equalizing two terms inside $\min$, and finally we use
$h(\rvY^n)\leq \nA n \log P +\taulog$. 
This establishes the converse proof. 

\subsection{Achievability proof of Theorem \ref{thm:SDoF}} \label{subsec:achievabilitySDoF}
In the following, we wish to show the achievability of the SDoF.
As in the converse part, we consider separately the cases for different
$m$. Note that only two ranges of $m$ need to be considered. The first
one is $\nB \le m\leq \max\{\nA,\nB\}$ and the other one is $\max\{\nA,\nB\} < m \le \nA+\nB$. 
For $m<\nB$, the SDoF is zero. For $m>\nA+\nB$, the converse shows that
it is useless in terms of SDoF to set more than $\nA+\nB$ antennas. 

\subsubsection{Case $\nB \le m\leq \max\{\nA,\nB\}$}
For this case, we need to show that $d_s(\nA,\nB,m)=m-\nB$ is achievable for $ \nB < m\leq  \nA$. 
This can be simply done by sending a vector of $m$ symbols
of which $m-\nB$ symbols $\rvV$ are useful message and the other $\nB$
symbols $\rvU$ are artificial noise (or a random message).
The legitimate receiver can decode all $m$ symbols and therefore extract
the useful message, i.e., 
\begin{align}
I(\rvV;\rvY) &= I(\rvV,\rvU;\rvY) - I(\rvU;\rvY \cond \rvV) \\
&= \log \det\left( \Id_{\nA} + \frac{P}{m} \Hm \Hm^{\H} \right) 
- \log \det\left( \Id_{\nA} + \frac{\nB P}{m} \Hm \left[
\begin{smallmatrix} \Id_{\nB} & \\ & 0_{m-\nB} \end{smallmatrix} \right]
  \Hm^{\H} \right) \label{eq:det1} \\
&\doteq (m - \nB) \log P,
\end{align}
while the eavesdropper channel is inflated by the
random message and does not expose more than a vanishing fraction of the
useful message, i.e., 
\begin{align}
I(\rvV;\rvZ) &= I(\rvV,\rvU;\rvZ) - I(\rvU;\rvZ \cond \rvV) \\
&= \log \det\left( \Id_{\nB} + \frac{P}{m} \Gm \Gm^{\H} \right) - \log \det\left( \Id_{\nB} + \frac{\nB P}{m} \Gm \left[
\begin{smallmatrix} \Id_{\nB} & \\ & 0_{m-\nB} \end{smallmatrix} \right] \Gm^{\H} \right) 
  \label{eq:det2} \\
  &\doteq \nB\log P - \nB\log P \\
&= 0,
\end{align}
where we used the fact that $\rank(\Hm)=m$ and $\rank(\Gm)=\nB$. Note that \eqref{eq:det1}
and \eqref{eq:det2} are obtained by applying independent Gaussian
signaling to $\rvV$ and $\rvU$ with proper covariance corresponding to
the power constraint. This assumption will be implicitly applied in the
rest of the paper.

\subsubsection{Case $\max\{\nA,\nB\} < m \le \nA + \nB$}

The proposed scheme combines the artificial
noise with the Maddah-Ali Tse~(MAT)
alignment scheme~\cite{maddah2010degrees}. The main idea of this scheme
is to send the artificial noise such that it  fills the
eavesdropper's observation and hides the confidential message, while it
shall be aligned in a reduced subspace at the legitimate receiver. 
\begin{table}[t]
\caption{Proposed three-phase scheme for $\max\{\nA, \nB\} < m\leq \nA +
\nB$.}
\centering
\begin{tabular}{|c|c|c|}
  \hline
  Phase 1 ($t\in\mathcal{T}_1$) & Phase 2 ($t\in\mathcal{T}_2$) & Phase 3 ($t\in\mathcal{T}_3$)\\ \hline
  $\xv_1 = \uv$ & $\xv_2 = \vv + \Thetam \tilde{\yv}_1$ & $\xv_3 = \Phim \tilde{\zv}_2$ \\ 
  $\tilde{\yv}_1 = \tilde{\Hm}_1 \uv$ & $\tilde{\yv}_2 = \tilde{\Hm}_2 \vv +
  \tilde{\Hm}_2 \Thetam \tilde{\Hm}_1 \uv $ & $\tilde{\yv}_3 = \tilde{\Hm}_3 \Phim \tilde{\Gm}_2 \vv +
  \tilde{\Hm}_3 \Phim \tilde{\Gm}_2 \Thetam \tilde{\Hm}_1 \uv$ \\ 
  $\tilde{\zv}_1 = \tilde{\Gm}_1 \uv$ & $\tilde{\zv}_2 = \tilde{\Gm}_2 \vv +
  \tilde{\Gm}_2 \Thetam \tilde{\Hm}_1 \uv$ & $\tilde{\zv}_3 = \tilde{\Gm}_3 \Phim \tilde{\Gm}_2 \vv +
  \tilde{\Gm}_3 \Phim \tilde{\Gm}_2 \Thetam \tilde{\Hm}_1 \uv$ \\ \hline
\end{tabular}
\label{table:scheme}
\caption{Length of three phases $\{\tau_i\}$ for different $m,\nA, \nB$.}
\centering
\begin{tabular}{|c|c|c|}
  \hline
 &  $\max\{\nA,\nB\} < m \leq \nA+\nB$ & $m> \nA+\nB $\\ \hline
$\tau_1$&$ \nA \nB$ & $\nA \nB$\\
 $\tau_2$& $\nA (m-\nB)$& $\nA^2$\\
$\tau_3$ &$(m-\nA)(m-\nB)$ & $\nA\nB$  \\ 
total duration $\sum_{i=1}^3 \tau_i$ &$\nA\nB + m(m-\nB)$ & $ 2\nA\nB +\nA^2$  \\ \hline
\end{tabular}
\label{table:Length}
\end{table}
The proposed three-phase scheme is presented in Table~\ref{table:scheme},
where the signal model without thermal noise is described concisely with the block matrix
notation:
\begin{flalign}
  &&\uv  &=  \bigl[\ \begin{matrix} \uv_1^\T & \dots & \uv_{\tau_1}^\T
  \end{matrix}\ \bigr]^\T \in \mathbb{C}^{m\tau_1\times1},
  & \vv  &=  \bigl[ \ \begin{matrix} \vv_1^\T & \dots & \vv_{\tau_2}^\T 
  \end{matrix}\ \bigr]^\T \in \mathbb{C}^{m\tau_2 \times1}, &&\\
  &&  \tilde{\Hm}_i &=
  \diag\left(\{\Hm_t\}_{t\in\mathcal{T}_i}\right)\in
  \mathbb{C}^{\nA\tau_i \times m\tau_i}, 
    & \tilde{\Gm}_i &= \diag\left(\{\Gm_t\}_{t\in\mathcal{T}_i}\right)
    \in \mathbb{C}^{\nB\tau_i\times m\tau_i},
     && \\ 
  &&   \Thetam &\in \mathbb{C}^{m\tau_2 \times \nA \tau_1}, 
  & \Phim &\in \CC^{ m \tau_3 \times \nB \tau_2},
     && 
\end{flalign}%
where $\tau_i$ denotes the length of phase $i$ for $i=1$, $2$, and $3$ given in Table \ref{table:Length}. 

The three phases are explained as follows: 
\begin{itemize}
  \item \emph{Phase 1, $t \in \mathcal{T}_1 \defeq
    \{1, \ldots, \tau_1\}$: \textbf{sending the artificial noise}.} The $m \tau_1$ symbols sent in
    $\tau_1$ time slots is represented by $\uv$. 
  \item \emph{Phase 2, $t \in \mathcal{T}_2 \defeq \{\tau_1+1, 
    \ldots, \tau_1+\tau_2\}$: \textbf{sending the confidential symbols
    with the artificial noise seen
    by the legitimate receiver}.} In $\tau_2$ time slots, we send the
    $m \tau_2$ useful symbols represented by $\vv$, superimposed by a
    linear combination~(specified by $\Thetam$) of the artificial noise observed by
    the legitimate receiver in phase 1.\footnote{As mentioned before, we
    ignore the scaling factor necessary to meet the power constraint. The same holds for the transmit vector in phase 3. } 
  \item \emph{Phase 3, $t \in \mathcal{T}_3 \defeq \{\tau_1+\tau_2+1, 
    \ldots, \tau_1+\tau_2+\tau_3\}$: \textbf{repeating the eavesdropper's observation during
    phase 2}.} The final phase consists in
    sending a linear combination~(specified by $\Phim$) of the
    eavedropper's observation in phase 2. The aim of this phase is to
    complete the equations for the legitimate receiver to solve the
    useful symbols $\vv$ without exposing anything new to the
    eavedropper.  
\end{itemize}

After three phases, the observations are given by
\begin{subequations}
\begin{align}
   \yv &= 
    \underbrace{\begin{bmatrix} 
      \Id_{\nA \tau_1}  & \zerov_{m\tau_2} \\
      \tilde{\Hm}_{2} \Thetam  &  \tilde{\Hm}_{2} \\ 
      \tilde{\Hm}_{3} \Phim \tilde{\Gm}_{2} \Thetam   & \tilde{\Hm}_{3} \Phim
      \tilde{\Gm}_{2}
    \end{bmatrix}}_{\He}
  \begin{bmatrix}  \tilde{\Hm}_{1} \uv \\
          \vv\end{bmatrix}  + \ev, \\
  \zv&= 
  \underbrace{\begin{bmatrix}
    \tilde{\Gm}_{1}  & \zerov_{\nB \tau_2} \\
    \tilde{\Gm}_{2} \Thetam \tilde{\Hm}_{1} & \Id_{\nB \tau_2}  \\
    \tilde{\Gm}_{3} \Phim \tilde{\Gm}_{2} \Thetam \tilde{\Hm}_{1} &
    \tilde{\Gm}_{3} \Phim
  \end{bmatrix}}_{\Ge}
   \begin{bmatrix} \uv \\
     \tilde{\Gm}_2\vv\end{bmatrix} +\bv.
\end{align}%
\end{subequations}
Therefore, we have
\begin{align}
  I(\rvV;\rvY) &= I(\rvV,\tilde{\Hm}_1\rvU;\rvY) - I(\tilde{\Hm}_1 \rvU;\rvY \cond \rvV) \\
&\doteq \rank\bigl( \He \bigr) \log(P) - 
\rank\left( 
\begin{smallmatrix} 
      \Id_{\nA \tau_1}  \\
      \tilde{\Hm}_{2} \Thetam  \\ 
      \tilde{\Hm}_{3} \Phim \tilde{\Gm}_{2} \Thetam   
  \end{smallmatrix} 
\right) \log(P)    \\
&= \left(\nA \tau_1+ \rank\left( 
\begin{smallmatrix} 
      \tilde{\Hm}_{2} \\ 
      \tilde{\Hm}_{3} \Phim \tilde{\Gm}_{2}    
  \end{smallmatrix} 
\right)\right) \log(P) -  \nA \tau_1 \log P  \label{eq:structure} \\
&= \rank\left( 
\begin{smallmatrix} 
      \tilde{\Hm}_{2} \\ 
      \tilde{\Hm}_{3} \Phim \tilde{\Gm}_{2}    
  \end{smallmatrix} 
\right)\log(P)  \\
&= m \nA (m-\nB) \log P,
\end{align}
where \eqref{eq:structure} follows due to the block-triangular structure
of $\He$ and by the fact that the rank of the second term corresponds to the rank of the identity matrix. In order to prove the
last equality, we need to show first that the submatrix $\tilde{\Hm}_{3} \Phim \tilde{\Gm}_{2}$ has full-row rank with linearly independent $\nA \tau_3$ rows. This is satisfied by letting 
\begin{align}\label{MatrixB}
  \Phim \pmb{\Pi} = \begin{bmatrix}
   \diag(\{ \Phim_t \}_{t=1}^{\tau_3})& \zerov_{m\tau_3 \times (\nB\tau_2-\nA \tau_3)}
   \end{bmatrix},
\end{align}
where $\pmb{\Pi}$ is a permutation matrix such that the first
$\nA\tau_3$ rows of $\pmb{\Pi}^\T \tilde{\Gm}_{2}$, is block diagonal, denoted by $\diag\left( \{\Gm_{2,t}^{\Pi}\}_{t\in\Tc_2}\right)$, 
where $\Gm_{2,t}^{\Pi}\in \CC^{(m-\nA)\times m}$ is a submatrix of $\tilde{\Gm}_{2,t}$;
$\Phim_t$ denotes a $m\times \nA$ matrix with $\nA$ independent
columns, e.g.,  $\Phim_t=\bigl[ \begin{matrix} \Id_{\nA} &
  \zerov_{\nA\times (m-\nA)}\end{matrix} \bigr]^\T$. Note that with this
  particular choice of $\Phim$, the resulting submatrix is given by 
\begin{align}
\tilde{\Hm}_{3} \Phim \tilde{\Gm}_{2} = \diag(\{
\Hm_{t+\tau_1+\tau_2}\Phim_t \}_{t=1}^{\tau_3}) \,
 \diag\left( \{\Gm_{2,t}^{\Pi}\}_{t\in\Tc_2}\right)\label{eq:tmp922}
\end{align}
Since the first matrix in the
right hand side of \eqref{eq:tmp922} is square and full-rank, it is easy
to see that $\rank\left( \begin{smallmatrix} \tilde{\Hm}_{2} \\
  \tilde{\Hm}_{3} \Phim \tilde{\Gm}_{2}    \end{smallmatrix} \right) =
  \rank\left( \begin{smallmatrix} \tilde{\Hm}_{2} \\
     \diag\left( \{\Gm_{2,t}^{\Pi}\}_{t\in\Tc_2}\right) \end{smallmatrix} \right) $.
By the row permutation, we can readily show that the latter has a desired rank of $m\nA (m-\nB)$. Namely, 
\begin{align}
\rank\left( \begin{matrix} \tilde{\Hm}_{2} \\
 \diag\left( \{\Gm_{2,t}^{\Pi}\}_{t\in\Tc_2}\right)  \end{matrix} \right) &= \sum_{t=1}^{\nA(m-\nB)} \rank\left(\begin{matrix} \tilde{\Hm}_{2,t} \\
    \tilde{\Gm}_{2,t}^{\Pi}  \end{matrix} \right)\\
    &= \nA(m-\nB) m,
\end{align}
where the last equality follows by noticing that each block $t$ corresponds to $m$ different rows of the state matrix $\Sm_t$ which are linearly independent from Assumption \ref{assumption:independency}. 
On the other hand, the eavesdropper's observation is filled by the artificial noise and thus does not expose more than a vanishing fraction of the
useful message, i.e., 
\begin{align}
  I(\rvV;\rvZ) &\le I(\tilde{\Gm}_2\rvV;\rvZ) \label{chain-rule}\\ 
  &= I(\tilde{\Gm}_2\rvV,\rvU;\rvZ) - I(\rvU;\rvZ \cond \tilde{\Gm}_2
  \rvV)\\
  &\doteq \nB (\tau_1+ \tau_2) \log P  - \rank\left( 
  \begin{smallmatrix}
    \tilde{\Gm}_{1}  \\
    \tilde{\Gm}_{2}\Thetam \tilde{\Hm}_{1} \\
    \tilde{\Gm}_{3} \Bm \tilde{\Gm}_{2} \Thetam \tilde{\Hm}_{1}  \end{smallmatrix}
  \right) \log P  \label{eq:rank-defficient2} \\
&= m \nA \nB \log P  - \rank\left( 
  \begin{smallmatrix}
    \tilde{\Gm}_{1}  \\
    \tilde{\Gm}_{2}\Thetam \tilde{\Hm}_{1} \end{smallmatrix}
  \right) \log P \label{repeat} \\
  &\doteq 0,
\end{align}
where \eqref{chain-rule} follows due to the Markov chain
$\rvV\leftrightarrow \tilde{\Gm}_2\rvV\leftrightarrow \rvZ$;
(\ref{eq:rank-defficient2}) follows by noticing that the rank of  $\Ge$
is determined by the submatrix corresponding to first two phases;
\eqref{repeat} follows because the third block row is a linear
combination of rows from the second block row. In order to prove the
last equality, we choose
\begin{align}\label{MatrixA}
  \Thetam \pmb{\Pi} = \begin{bmatrix}
   \diag(\{ \Thetam_t \}_{t=1}^{\tau_2})& \zerov_{m\tau_1 \times (\nA\tau_1-\nB \tau_2)}
   \end{bmatrix},
\end{align}
where $\pmb{\Pi}$ is a permutation matrix\footnote{We abuse the notation to denote another permutation matrix than the one used in \eqref{MatrixB}
.} such that the first
$\nB\tau_2$ rows of $\pmb{\Pi}^\T \tilde{\Hm}_{1}$ is block diagonal,
denoted by $\diag\left(\{\tilde{\Hm}_{1,t}^{\Pi}\}_{t\in\Tc_1}\right)$,
with $\tilde{\Hm}_{1,t}^{\Pi}$ being a  $(m-\nB)\times m$ submatrix of  $\tilde{\Hm}_{1,t}$;
$\Thetam_t$ denotes a $m\times \nB$ matrix with $\nB$ independent
columns, e.g.,  $\Thetam_t=\bigl[ \begin{matrix} \Id_{\nB} &
  \zerov_{\nB\times (m-\nB)}\end{matrix} \bigr]^\T$. By applying exactly
  the same reasoning as on the choice of $\Phim$, we can prove that     
$\rank\left( \begin{smallmatrix} \tilde{\Gm}_{1}  \\
  \tilde{\Gm}_{2}\Thetam \tilde{\Hm}_{1} \end{smallmatrix} \right) =
  m\tau_1= m\nA\nB$.  As a result, the $\nA m(m-\nB)$ useful symbols can be conveyed secretly over
$\nA\nB+m(m-\nB)$ time slots in the high SNR regime, yielding the SDoF of
$\frac{\nA m(m-\nB)}{\nA\nB+m(m-\nB)}$. 

\subsection{Achievability proof of Theorem \ref{thm:asymmetric}}
In this subsection, we provide the achievability proof for the case of
delayed partial CSIT when the transmitter has delayed CSI only on
the legitimate channel. We focus on the case $\max\{\nA,\nB\}< m <
\nA+\nB$. For the case of $m\geq \nA+\nB$, we can easily show that the
desired SDoF follows by using only $\nA+\nB$ antennas out of $m$, i.e.,
by replacing $m$ by $\nA+\nB$ similarly to the case of delayed CSIT on
both channels. We propose a variant of the artificial noise alignment scheme described
previously. The lack of CSIT on the eavesdropper channel requires the
following modifications. First, the transmission consists of first two
phases presented in Table \ref{table:scheme}, because the lack of CSI on
the eavesdropper channel does not enable the transmitter to repeat the
signal overheard by the eavesdropper (corresponding to the third phase).
Consequently, the confidential symbols $\vv$ sent during the second
phase must be decoded within this phase. This decreases the dimension of
$\vv$ from  $m\tau_2$ to $\nA \tau_2$. After two phases, the
observations are given by
\begin{subequations}
\begin{align}
   \yv &= 
    \underbrace{\begin{bmatrix} 
      \Id_{\nA \tau_1}  & \zerov_{\nA\tau_2} \\
      \tilde{\Hm}_{2} \Thetam  &  \tilde{\Hm}_{2} \\ 
     \end{bmatrix}}_{\He}
  \begin{bmatrix}  \tilde{\Hm}_{1} \uv \\
          \vv\end{bmatrix}  + \ev, \\
  \zv&= 
  \underbrace{\begin{bmatrix}
    \tilde{\Gm}_{1}  & \zerov_{\nB \tau_2} \\
    \tilde{\Gm}_{2} \Thetam \tilde{\Hm}_{1} & \Id_{\nB \tau_2}  \\  \end{bmatrix}}_{\Ge}
   \begin{bmatrix} \uv \\
     \tilde{\Gm}_2\vv\end{bmatrix} +\bv.
\end{align}%
\end{subequations}
Following similar steps as before and choosing $\Thetam$ in (\ref{MatrixA}), we can easily show that 
\begin{align}
  I(\rvV;\rvY) &\doteq  \nA^2 (m-\nB) \log P,  \\
  I(\rvV;\rvZ) &\doteq 0.
\end{align}
 As a result, the $\nA^2 (m-\nB)$ useful symbols can be conveyed secretly over
$\nA m $ time slots in the high SNR regime, yielding the SDoF of
$\frac{\nA (m-\nB)}{m}$.

\section{Broadcast Channel with Confidential Messages: Proof of Theorem \ref{thm:ubBCC}}\label{section:Proof-bcc}
\subsection{Converse}
We focus on the case $m>\max\{\nA,\nB\}$ in the following. The converse for the other cases is trivial from Section \ref{section:Proof-wiretap}. 
The secrecy constraint (\ref{eq:Constraint1-bcc}) together with Fano's 
inequality for $\WB$, i.e., $h(\WB|\zv^n)\leq n \epsilon$, yields 
\begin{equation} 
I(\WA;\rvZ^n|\WB)\leq n\, \taulog \label{eq:ubBCC6}.
\end{equation}
Similarly to the converse of the MIMO wiretap channel, we obtain two upper bounds on $R_{\textsf{A}}$. The first bound is obtained by
combining \eqref{eq:ubBCC6} with Fano's inequality on $\WA$, i.e., $h(\WA|\yv^n)\leq n \epsilon$, 
\begin{align}
n(R_{\textsf{A}}- \taulog) &\leq  I(\WA;\rvY^n|\WB) - I(\WA;\rvZ^n|\WB)\nonumber\\
&\leq  I(\WA;\rvY^n|\rvZ^n,\WB)\label{degrade}\\
 &\leq h(\rvY^n|\rvZ^n,\WB) \\
&\dotle \frac{\tilde{m}-\nB}{\nB}h(\rvZ^n|\WB), \label{apply-lemma}
\end{align}
where (\ref{degrade}) follows by $I(\WA;\rvY^n|\WB)\leq
I(\WA;\rvY^n,\rvZ^n|\WB)$; (\ref{apply-lemma}) follows from inequality
\eqref{eq:tmp1} in Lemma~1.
The second bound is
(\ref{eq:tmplast}) which holds also here by replacing $\rvW$ by $\WA$,
namely,
\begin{align}
I(\WA; \rvY^n)-I(\WA; \rvZ^n) \dotle h(\rvY^n) - \frac{\nA}{\tilde{m}} h(\rvZ^n). 
\end{align}
Putting the two upper bounds together, we have 
\begin{equation}\label{eq:ubBCC7}
 n(R_{\textsf{A}}- \taulog) \dotle  \min\left\{  \frac{\tilde{m}-\nB}{\nB}h(\rvZ^n|\WB), h(\rvY^n)- \frac{\nA}{\tilde{m}}h(\rvZ^n)\right\}. 
\end{equation} 

On the other hand, Fano's inequality for $\WB$ leads to
\begin{equation} 
  n(R_{\textsf{B}}- \taulog) \leq  h(\rvZ^n) - h(\rvZ^n|\WB). \label{eq:ubBCC8}
\end{equation}
Now, we sum inequalities (\ref{eq:ubBCC7}) and (\ref{eq:ubBCC8}) with the weight $T_{\textsf{A}}=\nA\nB + \tilde{m}(\tilde{m}-\nB)$, $\nA(\tilde{m}-\nB)$,
respectively. This yields 
\begin{align}
  n( T_{\textsf{A}} R_{\textsf{A}}+ \nA(\tilde{m}-\nB) R_{\textsf{B}}- \taulog)
  &\dotle\! \max_{h(\rvY^n)}\max_{\alpha}\min\left\{(\tilde{m}-\nB)\alpha, T_{\textsf{A}} h(\rvY^n)-\frac{\nA\nB}{\tilde{m}} \alpha\right\}\\
  &\leq \max_{h(\rvY^n)} \tilde{m}(\tilde{m}-\nB)h(\rvY^n)\\
  &\dotle\!  \nA \tilde{m}(\tilde{m}-\nB) n\log P,
\end{align}
where we let $\alpha= \nA h(\rvZ^n)+\frac{\tilde{m}(\tilde{m}-\nB)}{\nB}h(\rvZ^n|\WB)$ in the first inequality and the last inequality follows because 
$h(\rvY^n)\leq n \nA\log P+\taulog$. 
By dividing both sides by $\nA \tilde{m}(\tilde{m}-\nB) \log P$ and
letting $P$ grow, we obtain the first desired
inequality~\eqref{BCC-inq1}. Due to the symmetry of the problem,
\eqref{BCC-inq2} can be obtained by swapping the roles of $R_{\textsf{A}}$ and
$R_{\textsf{B}}$. This completes the converse proof.  

\subsection{Achievability}

\begin{table}
\caption{Proposed four-phase scheme for $\max\{\nA, \nB\} < m\leq \nA +
\nB$.}
\centering
\begin{tabular}{|c|c|c|}
  \hline
  Phase 1 & Phase 2 & Phase 3 \\ \hline
  $\xv_1 = \uv$ & $\xv_2 = \vA + \ThetamA \tilde{\yv}_1$ &  $\xv_3 =  \vB + \ThetamB \tilde{\zv}_1$  \\ 
  $\tilde{\yv}_1 = \tilde{\Hm}_1 \uv$ & $\tilde{\yv}_2 = \tilde{\Hm}_2
  \left( \vA + \ThetamA \tilde{\Hm}_1 \uv \right) $ &  $\tilde{\yv}_3 =
  \tilde{\Hm}_3 \left( \vB + \ThetamB \tilde{\Gm}_1 \uv \right) $ 
  \\ 
  $\tilde{\zv}_1 = \tilde{\Gm}_1 \uv$ & $\tilde{\zv}_2 = \tilde{\Gm}_2
  \left( \vA + \ThetamA \tilde{\Hm}_1 \uv \right) $ 
  & $\tilde{\zv}_3 = \tilde{\Gm}_3 \left( \vB + \ThetamB \tilde{\Gm}_1
  \uv \right) $ 
   \\ \hline
  \hline
  \multicolumn{3}{|c|}{  Phase 4 } \\ \hline
  \multicolumn{3}{|c|}{ $\xv_4 = \PhimA \tilde{\zv}_2 + \PhimB \tilde{\yv}_3$} \\ 
  \multicolumn{3}{|c|}{ $\tilde{\yv}_4 = \tilde{\Hm}_4 \PhimA( \tilde{\Gm}_2 \vA +
 \tilde{\Gm}_2 \ThetamA \tilde{\Hm}_1 \uv) +  \tilde{\Hm}_4\PhimB(\tilde{\Hm}_3 \vB +
  \tilde{\Hm}_3 \ThetamB \tilde{\Gm}_1 \uv)$ } \\ 
  \multicolumn{3}{|c|}{ $\tilde{\zv}_4 = \tilde{\Gm}_4 \PhimA (\tilde{\Gm}_2 \vA +
 \tilde{\Gm}_2 \ThetamA \tilde{\Hm}_1 \uv) +  \tilde{\Gm}_4\PhimB (\tilde{\Hm}_3 \vB +
  \tilde{\Hm}_3 \ThetamB \tilde{\Gm}_1 \uv)$} \\ \hline
\end{tabular}
\label{table:scheme-bcc}
\caption{Length of four phases $\{\tau_i\}$ for different $m$, $\nA$,
and $\nB$.}
\centering
\begin{tabular}{|c|c|c|}
  \hline
 &  $\max\{\nA,\nB\} < m \leq \nA+\nB$ & $m> \nA+\nB $\\ \hline
$\tau_1$&$ \nA \nB$ & $\nA \nB$\\
 $\tau_2$& $\nA (m-\nB)$& $\nA^2$\\
 $\tau_3$& $\nB (m-\nA)$& $\nB^2$\\ 
$\tau_4$ &$(m-\nA)(m-\nB)$ & $\nA\nB$  \\ 
total duration $\sum_{i=1}^4 \tau_i$ &$m^2$ & $ (\nA+\nB)^2$  \\ \hline
\end{tabular}
\label{table:Length2}
\end{table}

The corner points can be achieved by the ANA scheme described in Section \ref{section:Proof-wiretap}. Here, we provide a strategy achieving the sum SDoF point. 
In fact,  the ANA scheme for the MIMO wiretap channel in Section
\ref{section:Proof-wiretap} can be suitably modified to convey two
confidential messages. We focus on the case $\max\{\nA, \nB\} < m \leq \nA+\nB$ because the converse proof says that we only need to use $\nA+\nB$ antennas for the case $m> \nA+\nB$.  

The proposed four-phase scheme is presented in
Table~\ref{table:scheme-bcc},
where the signal model is describe concisely with the block matrix
notation:
\begin{flalign}
  &&\uv  &=  \bigl[\ \begin{matrix} \uv_1^\T & \dots & \uv_{\tau_1}^\T
  \end{matrix}\ \bigr]^\T \in \mathbb{C}^{m\tau_1\times1}, &&&&\\
  && \vA  &=  \bigl[ \ \begin{matrix} \vv_{\textsf{A},1}^\T & \dots &
    \vv_{\textsf{A},\tau_2}^\T 
  \end{matrix}\ \bigr]^\T \in \mathbb{C}^{m\tau_2 \times1}, & \vB  &=
  \bigl[ \ \begin{matrix} \vv_{\textsf{B},1}^\T & \dots &
    \vv_{\textsf{B},\tau_3}^\T 
  \end{matrix}\ \bigr]^\T \in \mathbb{C}^{m\tau_3 \times1}, &&\\
  &&  \tilde{\Hm}_i &=
  \diag\left(\{\Hm_t\}_{t\in\mathcal{T}_i}\right)\in
  \mathbb{C}^{\nA\tau_i \times m\tau_i}, 
    & \tilde{\Gm}_i &= \diag\left(\{\Gm_t\}_{t\in\mathcal{T}_i}\right)
    \in \mathbb{C}^{\nB\tau_i\times m\tau_i},
     && \\ 
     &&   \Thetam_{\textsf{A}} &\in \mathbb{C}^{m\tau_2 \times \nA \tau_1}, & 
    \Thetam_{\textsf{B}} &\in \mathbb{C}^{m\tau_3 \times \nB \tau_1}, &&
    \\
    && \Phim_{\textsf{A}} &\in \CC^{ m \tau_4 \times \nB \tau_2},  &
     \Phim_{\textsf{B}} &\in \CC^{ m \tau_4 \times \nA \tau_3},  && 
\end{flalign}
where the durations of four phases $\{\tau_i\}_{i}$ are given in Table \ref{table:Length2}. The four phases consist of: 
\begin{itemize}
  \item \emph{Phase 1, $t \in \mathcal{T}_1 \defeq
    \{1, \ldots, \tau_1\}$: \textbf{sending the artificial noise}.} The $m \tau_1$ symbols sent in
    $\tau_1$ time slots is represented by $\uv$. 
  \item \emph{Phase 2, $t \in \mathcal{T}_2 \defeq \{\tau_1+1, 
    \ldots, \tau_1+\tau_2\}$: \textbf{sending the confidential symbols
    $\vA$ with the artificial noise seen
    by receiver~\textsf{A}}.} In $\tau_2$ time slots, we send the
    $m \tau_2$ useful symbols represented by $\vA$, superimposed by a
    linear combination~(specified by $\Thetam_{\textsf{A}}$) of the
    artificial noise observed by
    receiver~\textsf{A} in phase 1.
  \item \emph{Phase 3, $t \in \mathcal{T}_3 \defeq \{\tau_1+\tau_2+1, 
    \ldots, \tau_1+\tau_2+\tau_3\}$: \textbf{sending the confidential symbols
    $\vB$ with the artificial noise seen
    by receiver~\textsf{B}}.} In $\tau_3$ time slots, we send the
    $m \tau_3$ useful symbols represented by $\vB$, superimposed by a
    linear combination~(specified by $\Thetam_{\textsf{B}}$) of the
    artificial noise observed by
    receiver~\textsf{B} in phase 1.
   \item \emph{Phase 4, $t \in \mathcal{T}_4 \defeq
     \{\tau_1+\tau_2+\tau_3+1, 
    \ldots, \tau_1+\tau_2+\tau_3+\tau_4\}$: \textbf{repeating the past
    observations during in phase 2 and 3}.} The final phase consists in
    sending a linear combination of 
    receiver~\textsf{B}'s observation in phase 2~(specified by
    $\Phim_{\textsf{A}}$) and receiver~\textsf{A}'s observation in phase 3~(specified by
    $\Phim_{\textsf{B}}$). The aim of this phase is to
    complete the equations for the intended receivers to solve the
    useful symbols without exposing anything new about the message to
    the unintended receivers. 
\end{itemize}

After four phases, the observations are given by
\begin{subequations}
\begin{align}
  \yv &= 
    \underbrace{\begin{bmatrix} 
      \tilde{\Hm}_2 & \tilde{\Hm}_2 \ThetamA  & \zerov\\ 
      \tilde{\Hm}_4 \PhimA \tilde{\Gm}_2&  \tilde{\Hm}_4 \PhimA \tilde{\Gm}_2 \ThetamA & \tilde{\Hm}_4\\
       \zerov & \Id_{\nA\tau_1} &  \zerov\\
       \zerov & \zerov & \Id_{\nA\tau_3} 
     \end{bmatrix} }_{\He_{\text{bcc}}}
  \begin{bmatrix}\vA \\  \tilde{\Hm}_1\uv \\   \tilde{\Hm}_3 \vB +
    \tilde{\Hm}_3 \ThetamB \tilde{\Gm}_1\uv\end{bmatrix} + \ev, \\
  \zv &= 
  \underbrace{\begin{bmatrix}
  \zerov & \Id_{\nB\tau_1} & \zerov \\
  \zerov & \zerov & \Id_{\nB\tau_2} \\
 \tilde{\Gm}_3 &\tilde{\Gm}_3 \ThetamB & \zerov	\\
 \tilde{\Gm}_4 \PhimB\tilde{\Hm}_3&  \tilde{\Gm}_4 \PhimB\tilde{\Hm}_3\ThetamB   &\tilde{\Gm}_4 \PhimA
    \end{bmatrix}}_{\Ge_{\text{bcc}}}
   \begin{bmatrix} \vB \\  \tilde{\Gm}_1\uv \\ 
   \tilde{\Gm}_2 \vA +\tilde{\Gm}_2  \ThetamA \tilde{\Hm}_1\uv \end{bmatrix} + \bv.
\end{align}%
\end{subequations}
First, we examine the mutual information between $\vA$ and $\yv$: 
\begin{align}
  I(\vA;\rvY) &= I(\vA,\tilde{\Hm}_1\rvU, \tilde{\Hm}_3 \vB +  \tilde{\Hm}_3 \ThetamB \tilde{\Gm}_1\uv;\rvY) - I(\tilde{\Hm}_1 \rvU,\tilde{\Hm}_3 \vB +  \tilde{\Hm}_3 \ThetamB \tilde{\Gm}_1\uv;\rvY \cond \vA) \\
  &\doteq \rank( \He_{\text{bcc}} ) \log(P) - 
\rank\left( 
\begin{smallmatrix} 
\tilde{\Hm}_{2} \ThetamA  & \zerov\\ 
      \tilde{\Hm}_4   \tilde{\Hm}_{3} \PhimA \tilde{\Gm}_{2} \ThetamA & \tilde{\Hm}_4  \\
      \Id_{\nA \tau_1} & \zerov  \\
      \zerov & \Id_{\nA \tau_3} 
  \end{smallmatrix} 
\right) \log P     \\
&= \left( \nA (\tau_1+\tau_3)+ \rank\left( 
\begin{smallmatrix} 
      \tilde{\Hm}_{2} \\ 
      \tilde{\Hm}_4\PhimA \tilde{\Gm}_{2}    
  \end{smallmatrix} 
\right) \right) \log P -  \nA( \tau_1 +\tau_3)\log P
\label{eq:structure2} \\
&= \rank\left( 
\begin{smallmatrix} 
      \tilde{\Hm}_{2} \\ 
      \tilde{\Hm}_{4} \PhimA \tilde{\Gm}_{2}    
  \end{smallmatrix} 
\right)\log P  \\
&= m\,\nA (m-\nB) \log P, \label{chooseB}
\end{align}
where in (\ref{eq:structure2}) the first term is due to the
block-triangular structure of $\He_{\text{bcc}} $ and the second term
follows because the rank corresponds to the rank of the identity matrix;
(\ref{chooseB}) follows by choosing $\PhimA$ given in (\ref{MatrixB})
where we replace $\tau_3$ with $\tau_4$.  

Next, in order to examine the leakage of $\vA$ to receiver~\textsf{B}, we write
\begin{align}
  I(\vA;\rvZ, \vB) &= I(\vA; \rvZ |\vB) \\
  &\le I(\tilde{\Gm}_2\vA;\rvZ | \vB) \label{chain}\\ 
  &= I(\tilde{\Gm}_2\vA,\rvU;\rvZ| \vB) - I(\rvU;\rvZ \cond \tilde{\Gm}_2 \vA, \vB) \label{chain2}\\
  &\le I( \tilde{\Gm}_1\uv, \tilde{\Gm}_2 \vA +\tilde{\Gm}_2  \ThetamA \tilde{\Hm}_1\uv;\rvZ| \vB) - I(\rvU;\rvZ \cond \tilde{\Gm}_2 \vA, \vB) \\
  &\doteq  \rank\left( 
  \begin{smallmatrix} 
\Id_{\nB\tau_1} & \zerov \\
\zerov & \Id_{\nB\tau_2} \\
\tilde{\Gm}_3 \ThetamB & \zerov	\\
\tilde{\Gm}_4 \PhimB\tilde{\Hm}_3\ThetamB   &\tilde{\Gm}_4 \PhimA
    \end{smallmatrix}
  \right) \log P    
    - \rank\left( 
  \begin{smallmatrix}
   \tilde{\Gm}_1  \\
   \tilde{\Gm}_2  \ThetamA \tilde{\Hm}_1\\
 \tilde{\Gm}_3 \ThetamB  \tilde{\Gm}_1\\
  \tilde{\Gm}_4 \PhimB\tilde{\Hm}_3\ThetamB   \tilde{\Gm}_1+ \tilde{\Gm}_4 \PhimA  \tilde{\Gm}_2  \ThetamA \tilde{\Hm}_1      
    \end{smallmatrix}
  \right) \log P  \label{rank-submatrices} \\
&= \nB (\tau_1+ \tau_2) \log P  - \rank\left( 
  \begin{smallmatrix}
   \tilde{\Gm}_1  \\
   \tilde{\Gm}_2  \ThetamA \tilde{\Hm}_1
    \end{smallmatrix}
  \right) \log P   \label{simplify}\\
  &\doteq 0,
\end{align}
where \eqref{chain} follows from the Markov chain $\vA \leftrightarrow
\tilde{\Gm}_2\vA \leftrightarrow \zv$; \eqref{chain2} is due to another Markov chain $(\tilde{\Gm}_2\vA,\rvU)\leftrightarrow( \tilde{\Gm}_1\uv, \tilde{\Gm}_2 \vA +\tilde{\Gm}_2  \ThetamA \tilde{\Hm}_1\uv)\leftrightarrow \zv$; in \eqref{simplify} 
we notice that two block columns of $\Ge_{\text{bcc}}$ is
block-triangular and the second term follows by keeping only linearly
independent block rows; the last equality is obtained by setting $\ThetamA$ given in \eqref{MatrixA}. 
As a result, the SDoF $\dA = \frac{\nA( \tilde{m}-\nB)}{ \tilde{m}}$ is
achieved with the proposed scheme. By symmetry of the problem, we have $\dB =\frac{\nB(\tilde{m}-\nA)}{ \tilde{m}}$ which completes the proof. 

\section{Conclusions and Perspectives}\label{section:Conclusions}
We studied the impact of delayed CSIT on the MIMO wiretap channel and
the MIMO broadcast channel with confidential messages (BCC) by focusing
on the secrecy degrees of freedom (SDoF) metric. The optimal SDoF region
of  the two-user Gaussian MIMO-BCC was fully characterized. It is shown
that an artificial noise alignment~(ANA) scheme, which can be regarded
as a non-trivial extension of Maddah-Ali Tse (MAT) scheme, can achieve
the entire SDoF region. The proposed ANA scheme enables to nicely
quantify the resource overhead to be dedicated to secure the
confidential messages, which in turn appears as a DoF loss. Although
delayed CSIT was found useful to improve the SDoF over a wide range of the MIMO system, 
our study somehow revealed the bottleneck of physical-layer security due
to its high sensitivity to the quality of CSIT.  

Several interesting open problems emerge out of this work. First, some
techniques used  for lower- and upper-bounding the SDoF in this work may
serve to enhance further insights on related problems for moderate SNR
regimes. Second, the characterization of the SDoF upper bound of the
Gaussian MIMO  wiretap channel with delayed partial CSIT remains open.
We emphasize that for the case of partial CSIT, the inequalities due to the channel symmetry still hold true,  but these do not seem to be enough to prove the converse. The challenge consists of finding novel and tighter inequalities that capture some new asymmetry between $h(\zv^n)$ and $h(\yv^n)$. Finally, the extension to more complex scenarios such as the BCC with more than two receivers can be also investigated.

\appendices 
\section{Proof of Lemma \ref{lemma:jointvssingle}} \label{appendix:lemma-jointvssingle}
\begin{lemma} \label{lemma:essential}
  Let $\rvx^L = \left( \rvx_1, \ldots, \rvx_L \right)$ be entropy-symmetric
  such that $h(\{x_j:\,j\in\mathcal{J}\}) =  h(\{x_k:\,k\in\mathcal{K}\})$, for any
  $|\mathcal{J}| = |\mathcal{K}| \le L$. Then, for any $M\ge N$, we have 
  \begin{subequations}
      \begin{align}
        h(\rvx^{N+k}) - h(\rvx^N) &\ge h(\rvx^{M+k}) - h(\rvx^M),
        \quad \forall\,k\ge0,
        \label{eq:ess1} \\
        \text{and} \quad M\, h(\rvx^N) &\ge N\, h(\rvx^M), \label{eq:ess2}
      \end{align}%
    \end{subequations}
  where we define $h(\emptyset) = 0$ for convenience of notation. 
\end{lemma}
\begin{proof}
  For $M=N$, the inequalities \eqref{eq:ess1} and \eqref{eq:ess2} hold
  with equality trivially. It is thus without loss of generality to
  assume that $M > N$. 
 
  We first prove inequality \eqref{eq:ess1}. It is readily shown that
  \begin{align}
    h(\rvx^{N+k}) - h(\rvx^{N}) &= h(\rvx_1,\ldots,\rvx_{N+k}) - h(\rvx_1,\ldots,\rvx_N) \\
    &= h(\rvx_1,\ldots,\rvx_{N+k}) - h(\rvx_{k+1},\ldots,\rvx_{N+k})
    \label{eq:tmp999}\\
    &= h(\rvx_1,\ldots,\rvx_{k} \cond \rvx_{k+1},\ldots,\rvx_{N+k}),
  \end{align}%
  where \eqref{eq:tmp999} is from the entropy-symmetry of $\rvx^L$.
  Since the last equality is decreasing with $N\ge0$, \eqref{eq:ess1} is
  immediate. 
  
  For the inequality \eqref{eq:ess2}, we prove it by induction on
  $L$. For $L=2$, the only non-trivial case is $M=2$ and $N=1$, where we have 
  \begin{align}
    2 h(\rvx_1) &= h(\rvx_1) + h(\rvx_2) \\
    &\ge h(\rvx_1,\rvx_2). 
  \end{align}%
  Assume that the result holds to $L=l-1$, i.e., \eqref{eq:ess2} is
  true for any $(M,N)\in\left\{ (j,k):\ l-1\ge j > k \right\}$. We would like to prove that it
  holds for any $(M,N)\in\left\{ (j,k):\ l \ge j > k \right\}$. In
  particular, all we need to prove is that the inequality holds for
  $M=l$ and any $N \le l-1$, i.e., 
  \begin{align}
    l \,h(\rvx^N) - N\,h(\rvx^l) &\ge 0, \quad\forall\,N\le l-1. \label{eq:induction}
  \end{align}%
  To this end, we first write
  \begin{align}
    l\,h(\rvx^N) - N\,h(\rvx^l) &= (l-N)\,h(\rvx^N)  - N\,(h(\rvx^l) -
    h(\rvx^N)). \label{eq:tmp000}
  \end{align}%
  For $N$ such that $l-N \le N$, we can lower-bound the right hand side~(RHS) of \eqref{eq:tmp000} as 
  \begin{align}
    (l-N) h(\rvx^N)  - N (h(\rvx^l) - h(\rvx^N)) &\ge (l-N) h(\rvx^N)  - N (h(\rvx^N) - h(\rvx^{2N-l})) \\
    &= N\, h(\rvx^{2N-l}) - (2N-l) h(\rvx^N)\\
    &\ge 0,
  \end{align}%
  where the first inequality is from the fact that applying
  \eqref{eq:ess1}, $h(\rvx^l) -
  h(\rvx^N) \le h(\rvx^N) - h(\rvx^{2N-l})$;
  the last inequality is from the induction assumption, since $(N,2N-l)$ is such
  that $l-1\ge N\ge 2N-l$. 
  
  For $N$ such that $l-N > N$, we lower-bound the RHS of
  \eqref{eq:tmp000} differently
  \begin{align}
    (l-N) h(\rvx^N)  - N (h(\rvx^l) - h(\rvx^N)) &\ge (l-N) h(\rvx^N)  - N \, h(\rvx^{l-N}) \\
    &\ge 0,
  \end{align}
  where the first inequality is from the fact that applying
  \eqref{eq:ess1}, $h(\rvx^l) - h(\rvx^N) \le h(\rvx^{l-N}) -
  h(\rvx^{0})$ with $h(\rvx^0) = h(\emptyset) = 0$ by definition; the last inequality
  is from the induction assumption since $(l-N,N)$ is such that $l-1 \ge l-N > N$.
  The proof for \eqref{eq:induction} in complete. 

\end{proof} 

By symmetry of the problem, we only need to prove \eqref{eq:tmp1}. We first consider the case $\nA+\nB\leq m$.
\begin{align}
  \nB h(\rvY^n ,\rvZ^n) &= \nB \sum_{t=1}^n h(\rvY_t,\rvZ_t \cond \rvY^{t-1}, \rvZ^{t-1})\\
  &= \nB \sum_{t=1}^n h(\pmb{\omega}_t \cond \rvY^{t-1}, \rvZ^{t-1})\\
  &\le (\nA+\nB) \sum_{t=1}^n h(\rvZ_t \cond \rvY^{t-1}, \rvZ^{t-1})
  \label{eq:tmp638} \\
  &\le (\nA+\nB) \sum_{t=1}^n h(\rvZ_t \cond \rvZ^{t-1}) \\
  &\le (\nA+\nB) h(\rvZ^{n}),
\end{align}
where we define $\pmb{\omega}_t = \{\pmb{y}_t, \pmb{z}_t\}$;
\eqref{eq:tmp638} is the application of \eqref{eq:tmp1}. 

When $m < \nA + \nB$, \eqref{eq:tmp638} is loose. We tighten the
bound as follows. 
\begin{align}
  \nB h(\rvY^n ,\rvZ^n) &= \nB \sum_{t=1}^n h(\pmb{\omega}_t \cond \rvY^{t-1}, \rvZ^{t-1})\\
  &= \nB \sum_{t=1}^n \left( h(\bar{\pmb{\omega}}_t \cond \rvY^{t-1}, \rvZ^{t-1}) + h(\hat{\pmb{\omega}}_t \cond \bar{\pmb{\omega}}_t, \rvY^{t-1}, \rvZ^{t-1}) \right) \\
  &\le \nB \sum_{t=1}^n h(\bar{\pmb{\omega}}_t \cond \rvY^{t-1},
  \rvZ^{t-1}) + \taulog \label{eq:tmp001}\\
  &\le m \sum_{t=1}^n h(\rvZ_t \cond \rvY^{t-1}, \rvZ^{t-1}) + \taulog \\
  &\le m \sum_{t=1}^n h(\rvZ_t \cond \rvZ^{t-1}) + \taulog \\
  &\le m h(\rvZ^{n}) + \taulog,
\end{align}
where we partition ${\pmb{\omega}}_t$ as ${\pmb{\omega}}_t = \{
\bar{\pmb{\omega}}_t,\ \hat{\pmb{\omega}}_t\}$ in such a way that
$\bar{\pmb{\omega}}_t$ and $\tilde{\pmb{\omega}}_t$ are of length $m$
and $\nA+\nB-m$, respectively; \eqref{eq:tmp001} is from the fact that
$h(\hat{\pmb{\omega}}_t \cond \bar{\pmb{\omega}}_t, \rvY^{t-1},
\rvZ^{t-1})\le h(\hat{\pmb{\omega}}_t \cond \bar{\pmb{\omega}}_t)$ 
and that $h(\hat{\pmb{\omega}}_t \cond \bar{\pmb{\omega}}_t)\le \taulog$  
with the same reasoning applied in \eqref{eq:tmp777}-\eqref{eq:tmp888}.

\end{document}